\documentclass[journal]{IEEEtran}
\usepackage{ifpdf}
\usepackage{cite}
\usepackage{algorithmic}
\usepackage{algorithm}
\ifCLASSINFOpdf
   \usepackage[pdftex]{graphicx}

\else
 \usepackage[dvips]{graphicx}
\fi
\usepackage[cmex10]{amsmath}
\usepackage{epstopdf}
\usepackage{array}
\usepackage{flushend}
\usepackage{balance}
\usepackage[top = 0.75 in, bottom = 0.5450in, left = 0.721181in, right = 0.721181in]{geometry}
\usepackage{eqparbox}

\usepackage[tight,footnotesize]{subfigure}

\usepackage{fixltx2e}

\usepackage{float}

\usepackage{dblfloatfix}
\usepackage{url}
\usepackage{cite}
\usepackage{amsthm}
\newtheorem{proposition}{Proposition}
\newtheorem{remark}{Remark}

\usepackage{amssymb}
\begin{document}

\title{DC-Bias and Power Allocation in Cooperative VLC Networks for Joint Information and Energy Transfer \thanks{Part of this work is presented at the IEEE Global Communications Conference (IEEE Globecom, Abu Dhabi, UAE, Dec. 9-13, 2018) \cite{globecom2}.}}

\setlength{\columnsep}{0.21 in}

\author{Mohanad~Obeed,~\IEEEmembership{Student Member,~IEEE,\thanks{M. Obeed, A. M. Salhab, and S. A. Zummo are with King Fahd University of Petroleum \& Minerals (KFUPM), Dhahran, Eastern Province, Saudi Arabia (email: g201106250@kfupm.edu.sa, salhab@kfupm.edu.sa, zummo@kfupm.edu.sa).}}
        Hayssam~Dahrouj,~\IEEEmembership{Senior Member,~IEEE,\thanks{H. Dahrouj is with Effat university, Makkah Province, Saudi Arabia (email: hayssam.dahrouj@gmail.com).}}
        Anas~M.~Salhab,~\IEEEmembership{Senior Member,~IEEE,}
        Salam~A.~Zummo,~\IEEEmembership{Senior Member,~IEEE,} and
        Mohamed-Slim~Alouini,~\IEEEmembership{Fellow,~IEEE} \thanks{M.-S. Alouini is with King Abdullah University of Science and Technology (KAUST), Thuwal, Makkah Province, Saudi Arabia (email: slim.alouini@kaust.edu.sa)}}
\maketitle

\begin{abstract}
Visible light communications (VLC) have emerged as strong candidates for meeting the escalating demand for high data rates. Consider a VLC network, where multiple access-points (APs) serve both energy-harvesting users (EHUs), i.e., users which harvest energy from light intensity, and information-users (IUs), i.e., users which gather data information. In order to jointly balance the achievable sum-rate at the IUs and the energy harvested by the EHUs, the paper considers maximizing a network-wide utility, which consists of a weighted-sum of the IUs sum-rate and the EHUs harvested-energy, subject to individual IU rate constraint, individual EHU harvested-energy constraint, and AP power constraints, so as to jointly determine the direct current (DC)-bias value at each AP, and the users’ powers. The paper solves such a difficult non-convex optimization problem using an iterative approach which relies on inner convex approximations, and compensates for the used approximations using proper outer-loop updates. The paper further considers solving the special cases of the problem, i.e., maximizing the sum-rate, and maximizing the total harvested-energy, both subject to the same constraints. Numerical results highlight the significant performance improvement of the proposed algorithms, and illustrate the impacts of the network parameters on the performance trade-off between the sum-rate and harvested-energy.
\end{abstract}
\IEEEpeerreviewmaketitle
\section{Introduction}

\subsection{Overview}
The need for high data rates and unregulated spectrum services has pushed the research community  to examine visible light communication (VLC)  techniques as a supplementary technology for indoor communication. This is especially the case because of the scarcity of the available radio-frequency (RF) spectrum due to ultra-dense network deployment. VLC technology uses the visible portion of the electromagnetic spectrum that is completely untapped, safe, free, and provides a high potential bandwidth for wireless data transmission, while also rejecting the existing RF interference \cite{survey_m, elgala2011indoor}. VLC, further, provides  larger energy efficiency, lower battery consumption, and smaller latency as compared to RF-based networks\cite{surveyvlc}. VLC can be indeed safely used in sensitive environments such as chemical plants, aircraft, and hospitals \cite{miyakoshi2013cellular}.  In spite of the small coverage of the transmitters in VLC systems, an exhaustive reuse of frequency can be implemented, with a relatively small effect on the performance due to the manageable  co-channel interference \cite{haas2016lifi}.
 Despite the aforementioned advantages, VLC networks can still be subject to several performance degradation factors, such as limited coverage, non-line-of-sight (non-LoS), failure transmission, frequent handover, and inter-cell
interference.

 Another attractive VLC feature of valuable interest is its energy harvesting capabilities, which are best enabled through equipping the VLC receivers with solar panels, so as to directly convert the light intensity into current signals without the need for external power supply \cite{wang2015design}, and with up to 40\% conversion efficiency \cite{king2012solar}. In practical indoor environments, however, two different types of users can typically co-exist, i.e., information-users (IUs) (such as mobiles, laptops, or tablets) and energy-harvesting users (EHUs) (such as Internet-of-things (IoT) devices, sensors, or relays). While IUs are data-hungry devices with specific data rate constraints, EHUs aim at harvesting visible light energy, which is especially feasible in indoor applications such as smart buildings, health monitoring, and sensors devices' applications. This motivates us in this paper to evaluate the benefit of a particular VLC-based scheme which considers the coexistence of both IUs and EHUs, and addresses the problem of jointly optimizing and balancing the achievable sum-rate at the IUs and the total harvested energy by EHUs, by means of adjusting the DC-bias at the VLC access-points and the powers of the users' messages.

\subsection{Related Work}
One of this paper goals is to optimize the VLC systems performance, a topic which is extensively studied in the literature of VLC systems, either by supplementing the network with additional RF APs \cite{JOCN, dyn, globecom, load2017}, or by applying  VLC APs cooperation \cite{user_hanzo, TGCN, coop, chen2013joint}.

The work considered in this paper is particularly related to the problems of jointly maximizing IUs sum-rate and EHUs harvested energy. In fact, maximizing the achievable sum-rate is investigated widely in the literature of wireless networks \cite{qian2009mapel, dahrouj2012power}, and in VLC systems \cite{jiang2017joint, shen2016rate}. Due to LEDs limited coverage and failure to operate in non-LoS environments, maximizing the sum-rate in VLC systems is intrinsically different from the one considered in RF-based systems, as VLC systems often impose additional systems constraints, e.g., handover overhead \cite{wu2017joint}, blockage probability \cite{Shadow}, users distribution in the floor area \cite{dyn},  users' field-of-view (FoV) alignments \cite{survey_m}, and fractional-frequency reuse \cite{chen2015fractional}. For instance, one of the effective solutions for addressing VLC coverage and blockage issues is considered in \cite{chen2013joint}, which proposes deploying multiple cooperative distributed APs so to simultaneously serve multiple users, a scenario that is partially adopted in this current paper which considers simultaneously serving both IUs and EHUs.

From the perspective of VLC harvesting energy capabilities, VLC devices often harvest energy from the received light intensity \cite{solar,dual1,dual2,sum_rate, amr,on3, simul2}. For example, reference \cite{solar} verifies experimentally the light energy harvested at the mobile phone, when equipped with a commercial solar panel in indoor environments. The authors in \cite{solar} show that the devices that are directly exposed to the indoor light can be charged up to a satisfactory power level. Reference \cite{fakidis2016indoor} studies the concept of indoor optical wireless power transfer to solar cells during the darkness hours. By using 42 laser diodes, an optical power of $7.2$ Watts can be delivered to a $30$ meters distant solar panel.

In the context of simultaneous power and information transfer, \cite{dual1} and \cite{dual2} study a dual-hop hybrid VLC/RF communication system, as a means to reach out to the out-of-coverage user. References \cite{dual1, dual2} show that visible light can be used in the first hop to transfer both data information and energy to the relay which, in turn, forwards the data to the destination using the resulting harvested energy. References \cite{sum_rate,amr} maximize the sum-rate utility of a VLC system consisting of one AP and $K$ users, subject to individual QoS constraints. Reference \cite{sum_rate} assumes that user $k$ receives the information in their assigned time slot, and receives the power within the time slots that are assigned for other users.  Reference \cite{amr}, on the other hand, proposes solving the problem of allocating the optical intensity and time slots through adopting a loose upper bound on the individual required harvested energy.
The authors in \cite{on3} characterize the outage performance of a hybrid VLC-RF system, where the visible light is used in the downlink direction to transfer energy and data to the users, and then the users use the harvested energy to transmit RF signals in the uplink direction.

 All the above papers use the alternating current (AC) component for harvesting the energy at users, where the direct current (DC) component of the transmitted light is fixed and readily used to harvest energy.
Toward this direction, reference \cite{simul2} considers a network that enables simultaneous light-wave information and power transfer (SLIPT) and maximizes the harvested energy under QoS constraints, so as to determine the DC in a portion of time, given that DC-bias in the remaining time is fixed for the purpose of transferring energy only. Reference \cite{simul2}, however, is restricted to a single transmitter and a single receiver only. The authors in \cite{simul2}, further, do not impose any energy constraints on users operation.


\subsection{Contributions}
 Different from the aforementioned references, this paper considers a VLC network, where multiple APs cooperate to serve both EHUs (e.g. sensors or IoT devices), and IUs (e.g. laptops, mobile phones, etc.), so as to best capture the multi-diverse applications schemes expected in next generations of wireless networks. The paper then investigates the problem of balancing the achievable sum-rate at the IUs and the total harvested energy by the EHUs, by means of adjusting the DC-bias at APs and allocating users' powers.  For mathematical tractability, the paper adopts the zero-forcing (ZF) precoding approach to cancel intra-cell interference, similar to \cite{hanzo_haas, ICC}.

 To balance between the performance of the IUs and the EHUs, the paper formulates the optimization problem which maximizes a weighted sum of the IUs sum-rate and the EHUs total harvested energy, under QoS and illumination constraints. The performance of the system is a function of both the DC bias values allocated at each AP, and the powers assigned to the users' messages. One of the paper contributions is to solve such a difficult non-convex optimization problem using an iterative approach, which uses inner convex approximations of the objective and constraints. It then compensates for the approximations using proper outer-loop updates. The paper also proposes a simpler sub-optimal baseline approach, which provides a feasible, yet  simple, solution to the formulated problem based on equal DC-bias allocation. The paper further considers solving the two special cases of the original optimization, i.e., the problem of maximizing the IUs sum-rate, and the problem of maximizing the EHUs total harvested energy, both subject to the same constraints as above. Simulation results highlight the performance and the convergence of our proposed algorithms. They particularly suggest that appreciable harvested energy and sum-rate improvement can be reached by optimizing the DC-bias and messages' powers in VLC systems.

The reminder of this paper is organized as follows. The system model, VLC channel and the energy harvesting channel are presented in
Section II. In Section III, we formulate the problem and present the proposed algorithms that solve the formulated problem. We introduce and discuss simulation results in Section IV. The paper is then concluded in Section V.

\section{System and Channel Models}

\subsection{System Model}

\begin{figure}[!t]
\centering
\includegraphics[width=0.40\textwidth]{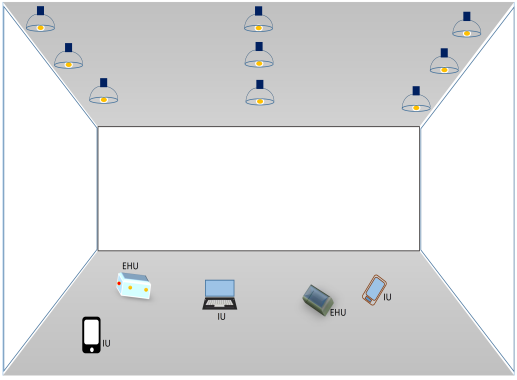}
\caption{System model (an example of user distribution when $N_{u,1}=3$, $N_{u,2}=2$, and $N_A=9$).}
\label{SM}
\end{figure}

Consider an indoor VLC system consisting of $N_A$ VLC access points (APs), which serve $N_u$ users in total. Among the $N_u$ users, $N_{u,1}$ users are IUs, and $N_{u,2}$ are EHUs, i.e., $N_u=N_{u,1}+N_{u,2}$. The paper considers the case where $N_{u,1}<N_A$, and adopts a zero-forcing (ZF) approach to cancel the intra-cell interference, so as to simplify the mathematical tractability of the problem. This assumption, i.e., $N_{u,1}<N_A$, emulates several indoor environments where the number of IUs are less than the number of lamps in the ceiling such as offices, labs, companies ,houses, etc.  Let $\mathbf{s}\in \mathbb{R}^{N_{u,1}\times 1}$ be the vector hosting the information of the $N_{u,1}$ users, and let $\mathbf{G}\in \mathbb{R}^{N_A\times N_{u,1}}$ be the precoding matrix associated with $\mathbf{s}$. The matrix $\mathbf{G}$ can be written as $\mathbf{G}=[\mathbf{g}^T_1;\ldots;\mathbf{g}_{N_A}^T],$ where $\mathbf{g}_i\in \mathcal{R}^{N_{u,1}\times 1}$ is the $i^{th}$ row of matrix $\mathbf{G}$ and the symbol $(.)^T$ means the transpose of the matrix. The electrical signal transmitted from AP $i$ can be written as $x_i=\mathbf{g}_i\bar{\mathbf{P}}\mathbf{s}$, where $\bar{\mathbf{P}}=diag([\sqrt{P_1},\sqrt{P_2},\ldots,\sqrt{P_{N_{u,1}}}])$,  where $P_j$ is the power allocated to the $j$th user's message. The DC-bias is denoted by $b_i$ and must be added to $x_i$ to avoid the resulting non-positive signals \cite{simul2}. The electrical signal, afterwards, modulates the optical intensity of the light-emitting diodes (LEDs) at AP $i$. The transmitted signal at AP $i$ can, therefore, be written as:
\begin{equation}
\label{transmit_signal}
y_{t,i}=P_{opt}(b_i+x_i),
\end{equation}
where $P_{opt}$ is the LEDs power at each AP. Let $I_L$ and $I_H$ be the minimum and the maximum input bias currents, respectively, i.e., $b_i \in [I_L,I_H]$. To guarantee that the output optical power is a linear function of the input current,  the transmitter LED must be in its  linear region. To this end, the peak amplitude of the modulated signal $x_i$, denoted by $A_i$, must satisfy the following constraint:
\begin{equation}
\label{Ai1}
A_i\leq \min(b_i-I_L,I_H-b_i).
\end{equation}
Constraint (\ref{Ai1}) implies that $A_i$ must satisfy two constraints, which are $A_i+b_i\leq I_H$ and $b_i-A_i\geq I_L$, to guarantee that the input electrical current to the LED is within the range of the linear region LED operation.
\subsection{Channel Model}
The paper adopts a line-of-sight (LoS) VLC channel model as in \cite{wir} and the first order reflected path for simplicity. More precisely, the LoS link between the $i$th LED and the $j$th user, denoted by $h_{i,j}^{LoS}$, can be written as follows:

\begin{equation}
\label{vlcch}
h_{i,j}^{LoS}=\frac{(m+1)A_{p}}{2\pi d_{i,j}^2} \cos^m(\phi)g_{of}f(\theta)\cos(\theta),
\end{equation}%
where $m=-1/\log_2(\cos(\theta_{\frac{1}{2}})$ is the Lambertian index, where $\theta_{\frac{1}{2}}$ is the half intensity radiation angle, $A_p$ is the photo-diode (PD) physical area, $d_{i,j}$ is the distance between the AP $i$ and the  user $j$, $g_{of}$ is  the optical filter gain, $\phi$ is the radiance angle, $\theta$ is the incidence angle, and $f(\theta)$  is the optical concentrator gain given by:
\begin{equation}
\label{FoVE}
f(\theta)= \begin{cases}\frac{n^2}{\sin^2(\Theta)}, & \theta\leq\Theta; \\
0, & \theta>\Theta,
\end{cases}
\end{equation}
where $\Theta$ is the semi-angle
of the user's FoV and $n$ is the refractive index.
The reference \cite{komine2004fundamental} shows that the DC  attenuation from the first reflected link is given by:

\begin{equation}
\begin{split}
h_{i,j}^{NLoS}= &\frac{(m+1)A_{p}}{2\pi d_{k,i}^2  d_{j,k}^2}. \rho. dA_s .\cos^m(\phi_r).cos(\alpha_1).cos(\alpha_2)\\
&.g_{of}f(\theta_r).\cos(\theta_r),
\end{split}
\end{equation}
where $\theta_r$ and $\phi_r$ are the incidence and irradiance angles of the first reflected link, respectively, where $d_{k,i}^2$ and $d_{j,k}^2$ are the distance between the $i$th AP and the $k$th reflecting point and the distance between the $k$th reflecting point and the $j$th user, respectively, where  $dA_s$ and $\rho$  are the reflective area  and the reflection factor, respectively, and where $\alpha_1$ and $\alpha_2$ are the the irradiance angles with respect to the reflecting point and with respect to the PD, respectively. Therefore, the equivalent VLC channel between $i$th AP  and $j$th user can be expressed as:
\begin{equation}
h_{i,j}=h_{i,j}^{LoS}+h_{i,j}^{NLoS}.
\end{equation}


After removing the DC-bias at the receiver side, the received signal vector at the users from all APs is given by:
\begin{equation}
\label{Yrr}
\mathbf{Y}_{r}=\rho \mathbf{H}\mathbf{G}\bar{\mathbf{P}} \mathbf{s}+\mathbf{n},
\end{equation}
where $\rho$ is the optical-to-electric conversion factor,  $\mathbf{H} \in \mathbb{R}^{N_{u}\times N_{A}}$ is the channel attenuation matrix that is assumed to be known at APs,  and $\mathbf{n}$ is the noise vector which includes both the thermal noise and the shot noise at the user, which can be modeled as zero-mean real-valued AWGN with variance $\sigma^2 = N_0W$, where $W$ is the modulation bandwidth, and $N_0$ is the noise power spectral density. The precoding matrix $\mathbf{G}$ is used to diagonalize the channel matrix, i.e.,  $\mathbf{G}=\mathbf{H}^T(\mathbf{H}\mathbf{H}^T)^{-1}$.

A tight lower bound on the network sum-rate at the $N_{u,1}$ IUs can then be written as \cite{tight}:

\begin{equation}
\label{sum_rate}
f_R(\mathbf{P})= \beta\sum_{j=1}^{N_{u,1}}\log\left(1+\frac{e(\rho^2)P_{j}}{2\pi WN_0}\right),
\end{equation}
where $\beta=W/2$ is a constant and $e$ is the constant exponential (Euler's number).
It is important to note that, from  (\ref{Yrr}), we can define the relation between the transmit power at AP $i$ and the assigned powers of the messages as:

\begin{equation}
\label{Pi_Pj}
p_i=\sum_{j=1}^{N_{u,1}}g_{i,j}^2P_j,
\end{equation}
where $g_{i,j}$ is the $(i,j)$th element of matrix $\mathbf{G}$. Based on  (\ref{transmit_signal}), the transmit power at the $i$th AP is further related to the peak amplitude of the modulated signal as follows:
\begin{equation}
\label{P_i}
p_i=(P_{opt}A_i)^2.
\end{equation}

\subsection{Energy Harvesting Signals}

For the EHUs, the DC component of the received signal is blocked by a capacitor and forwarded to the energy harvesting circuit \cite{simul2}. The harvested energy is given by \cite{solar}:

\begin{equation}
E=fI_{DC}V_{oc},
\end{equation}
where $f$ is the fill factor (typically around 0.75), and $I_{DC}$ is the received DC current which is given at the $j$th user by:
\begin{equation}
I_{DC}=\rho P_{opt} \mathbf{h}_{j}^T\mathbf{b},
\end{equation}
where $\mathbf{b}=[b_1, b_2,\ldots,b_{N_A}]$ is the DC-bias vector at APs, $\mathbf{h}_{j}$ is the channel vector from all APs to the user $j$, and
\begin{equation}
V_{oc}= V_t \ln(1+\frac{I_{DC}}{I_0}),
\end{equation}
where $V_t$ is the thermal voltage and $I_0$ is the dark saturation current of the PD. Hence, the harvested energy at user $k$ is given by:
\begin{equation}
\label{EH}
E_k(\mathbf{b})=f \rho P_{opt} V_t  \mathbf{h}_{k}^T\mathbf{b}\ln \left(1+\frac{\rho \mathbf{h}_{k}^TP_{opt}\mathbf{b}}{I_0}\right),
\end{equation}
and the total harvested energy at all $N_{u,2}$ users is given by

\begin{equation}
\label{THE}
f_E(\mathbf{b})=\sum_{k=1}^{Nu_2}E_k(\mathbf{b}).
\end{equation}

It is important to note that the DC-bias $b_i$ at the $i$th AP must be greater than or equal to $\frac{I_H+I_L}{2}$. This is because decreasing $b_i$ to be less than $\frac{I_H+I_L}{2}$ results in decreasing the harvested energy (\ref{EH}). It also decreases $A_i$ (based on (\ref{Ai1})), which decreases the transmit power $p_i$ that leads to a decrease in the sum-rate (\ref{sum_rate}). $b_i$, therefore, should satisfy: $b_i\geq \frac{I_H+I_L}{2}$, and $\min(I_H-b_i,b_i-I_L)$ becomes equal to $I_H-b_i$. Therefore, the relation  between $A_i$ and $b_i$ provided in (\ref{Ai1}) becomes:
\begin{equation}
\label{Aii}
A_i\leq I_H-b_i.
\end{equation}
In addition, if the optimal $b_i$ satisfies $b_i \leq I_H-A_i$, it can be increased to have $b_i=I_H-A_i$, which increases the objective function without violating the QoS constraints. Hence, the inequality in (\ref{Aii}) is satisfied with equality:
\begin{equation}
\label{Ai}
A_i= I_H-b_i.
\end{equation}

\section{Problem Formulation and Algorithms}

In order to jointly optimize the achievable sum-rate utility at the IUs and the total harvested energy utility at the EHUs, this section motivates and then considers maximizing a weighted sum of both utilities under QoS constraints and maximum transmit power constraint. The section then proposes two different solutions to solve the formulated non-convex problem by efficiently  adjusting the DC-bias vector and the users' powers. The section finally addresses the two individual optimization problems separately, i.e., maximizing the sum-rate utility, and then maximizing the total harvested energy utility, both under the same constraints.

\subsection{Weighted Sum Utility Maximization}
   The utility function for the IUs is the sum-rate that is given in (\ref{sum_rate}), which is a function of the messages' powers, while the utility function for the EHUs, given in (\ref{THE}), is the total harvested energy, which is a function of the APs DC-bias. Using (\ref{Pi_Pj}), (\ref{P_i}), and (\ref{Ai}), the relation between both variable vectors (the messages' powers and the DC-bias) can be expressed as follows:

 \begin{equation}
 \label{Pjbi}
 \sum_{j=1}^{N_{u,1}}g_{i,j}^2P_j=P_{opt}^2(I_H-b_i)^2, \ \ i=1,\ldots,N_A.
 \end{equation}
Expression (\ref{Pjbi}) shows that the relation between both vectors is not one-to-one. More specifically, a unique DC-bias vector $\mathbf{b}$ can be calculated for a given messages' power vector. The messages' power vector $\mathbf{P}$ might, however, have several solutions from a given DC-bias vector. Expression (\ref{Pjbi}) also shows that increasing the DC-biases  increases the total harvested energy at the EHUs, but decreases the data rate at the IUs (because the transmit power for the information signal at AP $i$ is given by $p_i=(P_{opt}A_i)^2$). Such conflicting impact of the DC-bias motivates the need for jointly optimizing both utilities by means of maximizing a weighted-sum under QoS and LEDs' linear operational region constraints. In this formulated problem, the weights of the utility functions can be controlled by a variable, called $\alpha \in [0,1]$. Mathematically, the considered optimization problem can be formulated as follows:

 \begin{subequations}
\label{EHM}
\begin{eqnarray}
&\displaystyle\max_{\mathbf{b}, \mathbf{P}}& \alpha f_R(\mathbf{P})+\frac{(1-\alpha)}{\omega} f_E(\mathbf{b})\\
\label{EHMb}
\nonumber
&s.t.&   \beta\log\left(1+\frac{e(\rho^2)P_j}{2\pi WN_0}\right)\geq R_{th,j}, \\&&\ \ \ \ \ j=1,\ldots,N_{u,1}\\
\label{EHMc}
\nonumber
&&   f \rho P_{opt} V_t \mathbf{h}_{k}^T\mathbf{b} \ln(1+\frac{\rho P_{opt} \mathbf{h}_{k}^T\mathbf{b}}{I_0})\geq E_{th,k},\\&&\ \ \ \  k=1,\ldots,N_{u,2}\\
\label{EHMd}
&& \sum_{j=1}^{N_{u,1}}g_{i,j}^2P_j= P_{opt}^2 (I_H-b_i)^2,\ \ i=1,\ldots,N_A,\\
\label{EHMe}
&&\frac{I_H+I_L}{2}\leq b_i\leq I_H, \ \ i=1,\ldots,N_A,
\end{eqnarray}
\end{subequations}
 where $R_{th,j}$ and $E_{th,k}$ are the minimum required data rate at the $j$th user and the minimum required energy to be harvested by the $k$th user, respectively, and  $\omega$ is a constant chosen to numerically equalize the order of magnitudes of the functions $f_R(\mathbf{P})$ and $\frac{1}{\omega} f_E(\mathbf{b})$. Constraints (\ref{EHMb}) and (\ref{EHMc}) are imposed to satisfy the minimum required fairness among IUs and the EHUs, while constraints in (\ref{EHMd}) and (\ref{EHMe}) are imposed to guarantee that the LEDs are operating in their linear region. It is important to note that problem (\ref{EHM}) solves three types of problems: 1) maximizing the sum-rate, which is achieved when we set $\alpha=1$, 2) maximizing the total harvested energy, which can be achieved by setting $\alpha=0$, 3) and maximizing a weighted sum of both utility functions for any $\alpha \in (0,1)$.


Problem (\ref{EHM}) cannot be easily solved, since the objective function and the constraint (\ref{EHMc}) are not concave, resulting in a difficult non-convex optimization problem. Specifically, $f_R(\mathbf{P})$ is a concave function in terms of $\mathbf{P}$, while $f_E(\mathbf{b})$ is a convex function in terms of $\mathbf{b}$, which makes their weighted sum a non-concave objective function. This paper next solves problem (\ref{EHM}) by first reformulating the problem in a more compact form, and then by proposing an numerical iterative approach.

The main idea of the proposed approach is that the problem is first formulated in terms of the messages' power vector $\mathbf{P}$ only, using the relation given in (\ref{Pjbi}). The paper then proposes a heuristic, yet efficient, algorithm to solve the reformulated problem through considering an approximated convex version of the problem, and then by correcting for the approximation in an outer loop update. For the sake of comparison, the paper further proposes a simple baseline approach, which guarantees a feasible solution to (\ref{EHM}).

\subsection{Problem Reformulation}

As discussed earlier, a unique DC-bias vector $\mathbf{b}$ can be calculated for a given messages' power vector. Thus, to reformulate problem (\ref{EHM}) in a more compact fashion, we choose to formulate the objective function and constraints of problem (\ref{EHM}) in terms of the vector $\mathbf{P}$ only. Using the relation in (\ref{Pjbi}), the DC-bias vector can be expressed as:

\begin{equation}
\label{P2b}
\mathbf{b}=I_H \mathbf{1}_{N_A}-\frac{1}{P_{opt}}\sqrt{\mathbf{\bar{G}P}},
\end{equation}
where $\mathbf{1}_{N_A}$ is the vector of length $N_A$ with all entries set to 1, where the matrix $\bar{\mathbf{G}}$ is defined as $\bar{\mathbf{G}}=[\bar{\mathbf{g}}_1^T; \bar{\mathbf{g}}_2^T;\ldots;\bar{\mathbf{g}}_{N_A}^T]$, with $ \bar{\mathbf{g}}_i^T=[g_{i,1}^2,\ g_{i,2}^2, \ldots, g_{i,N_{u,1}}^2]$, and where the square root denotes the componentwise square root of the vector argument.

Plugging  (\ref{P2b}) in the energy harvesting functions (\ref{EH}) and (\ref{THE}), we get:
\begin{equation}
\label{EH2}
\begin{split}
E_k(\mathbf{P})=f \rho P_{opt} V_t  \mathbf{h}_{k}^T(I_H \mathbf{1}_{N_A}-\frac{1}{P_{opt}}\sqrt{\mathbf{\bar{G}P}})\ln (   1+\\
\frac{\rho P_{opt} \mathbf{h}_{k}^T(I_H \mathbf{1}_{N_A}-\frac{1}{P_{opt}}\sqrt{\mathbf{\bar{G}P}})}{I_0}),
\end{split}
\end{equation}
and
\begin{equation}
\label{THE2}
f_E(\mathbf{P})=\sum_{k=1}^{N_{u,2}}E_k(\mathbf{P}).
\end{equation}
Using (\ref{P2b}), the constraints in (\ref{EHMe}) can be rewritten as:
\begin{equation}
\label{bcons}
0 \leq \mathbf{\bar{g}}_i^T\mathbf{P}\leq P_{opt}^2(\frac{I_H-I_L}{2})^2,  \ \ i=1,\ldots,N_A.
\end{equation}

Substituting  (\ref{EH2}), (\ref{THE2}), and (\ref{bcons}) in the optimization problem (\ref{EHM}), the problem can then be formulated in terms of the messages' power vector as follows:
 \begin{subequations}
\label{EHM2}
\begin{eqnarray}
&\displaystyle\max_{\mathbf{P}}& \alpha f_R(\mathbf{P})+\frac{(1-\alpha)}{\omega} f_E(\mathbf{P})\\
\label{EHM2b}
&s.t.&   P_j\geq P_{j,min}, \ \ \ \ \ j=1,\ldots,N_{u,1}\\
\label{EHM2c}
&&   E_k(\mathbf{P})\geq E_{th,k},\ \ \ \  k=1,\ldots,N_{u,2}\\
\label{EHM2d}
&&\mathbf{\bar{g}}_i^T\mathbf{P}\geq 0,  \ \ i=1,\ldots,N_A,\\
\label{EHM2e}
&& \mathbf{\bar{g}}_i^T\mathbf{P}\leq P_{opt}^2(\frac{I_H-I_L}{2})^2,\  \resizebox{0.11\textwidth}{!}{$i=1,\ldots,N_A,$}
\end{eqnarray}
\end{subequations}
where $P_{j,min}= \frac{(2^{\frac{R_{th,j}}{\beta}}-1)2\pi WN_0}{P_{opt}^2 e\rho^2}$.
Because functions $f_E(\mathbf{P})$ and $E_k(\mathbf{P})$ are not concave, the problem in (\ref{EHM2}) is still a non-convex optimization problem. Hence, we next propose a novel method that solves problem (\ref{EHM2}) by using a proper convex approximation, and then by compensating for the approximation in the outer loop.

\subsection{Problem Convexification}
To convexify problem (\ref{EHM2}), we utilize a two-step iterative approach. At the first step,
we fix the DC-bias vector values for specific terms of the non-concave functions, so as to get rid of the square root and the logarithm function expression in the energy functions. After solving the problem, the second step substitutes the updated value of the DC-bias vector in the terms of the non-concave functions. More specifically, in the first step (and at the very first iteration), let $\mathbf{\hat{b}}=\frac{I_H+I_L}{2}\mathbf{1}_{N_A}$ (i.e., $\hat{b}_i=\frac{I_H+I_L}{2}$) be the initial DC-bias vector. Therefore, the relation in (\ref{Pjbi}) can be approximated as follows:

\begin{equation}
 \label{Pjbi_apr}
 \mathbf{\bar{G}}\mathbf{P} \cong P_{opt}^2(I_H\mathbf{1}_{N_A}-\mathbf{b})\odot(I_H\mathbf{1}_{N_A}-\mathbf{\hat{b}}),
 \end{equation}
 where the operator $\odot$ denotes the vector componentwise multiplication. The DC-bias vector can be approximated as follows:

 \begin{equation}
\label{P2b_apr}
\mathbf{b} \cong I_H \mathbf{1}_{N_A}-\frac{1}{P_{opt}}\mathbf{\bar{G}P}./(I_H\mathbf{1}_{N_A}-\mathbf{\hat{b}}) ,
\end{equation}
where $./$ means the vector componentwise division. Define $$\mathbf{G}_b=\frac{1}{P_{opt}}[\frac{1}{I_H-\hat{b}_1}\bar{\mathbf{g}}_1^T; \frac{1}{I_H-\hat{b}_2}\bar{\mathbf{g}}_2^T;\ldots;\frac{1}{I_H-\hat{b}_{N_A}}\bar{\mathbf{g}}_{N_A}^T],$$ we can re-write (\ref{P2b_apr}) as follows:
\begin{equation}
\label{P2b_apr2}
\mathbf{b}\cong I_H \mathbf{1}_{N_A}-\mathbf{G}_b\mathbf{P}.
\end{equation}

To further convexify the energy functions, define $z_k(\mathbf{\hat{b}})$ as $z_k(\mathbf{\hat{b}})= \ln(1+\frac{\rho P_{opt}\mathbf{h}_{k}^T\mathbf{\hat{b}}}{I_0})$, which is a constant that depends on $\mathbf{\hat{b}}$.

Problem (\ref{EHM2}) can now be readily approximated as a convex optimization problem. For the completeness of presentation, define the following variables (which are all functions of the estimated DC-bias vector $\mathbf{\hat{b}}$):\\
$\mathbf{x}_k=f \rho P_{opt} V_t z_k(\mathbf{\hat{b}}) \mathbf{h}_{k},\ k=1,\ldots,N_{u,2}$,\\
 $x=\sum_{k=1}^{N_{u,2}}\mathbf{x}_k^T(I_H\mathbf{1}_{NA})$,\\
 $\mathbf{w}=\sum_{k=1}^{N_{u,2}}\mathbf{x}_k^T\mathbf{G}_b$, $m_k=I_H \mathbf{x}_k\mathbf{1}_{N_A}-E_{th,k}$, and \\
 $\mathbf{w}_k=\mathbf{x}_k^T\mathbf{G}_b, \ k=1,\ldots,N_{u,2}$.

Using the above notations, problem (\ref{EHM2}) can be approximated as follows:

\begin{subequations}
\label{EHM3}
\begin{eqnarray}
&\displaystyle\max_{\mathbf{P}}& \alpha f_R(\mathbf{P})+\frac{(1-\alpha)}{\omega}(x-\mathbf{w}^T\mathbf{P})\\
\label{EHM3b}
&s.t.&   P_j\geq P_{j,min}, \ \ \ \ \ j=1,\ldots,N_{u,1}\\
\label{EHM3c}
&&   \mathbf{w}_k^T\mathbf{P}\leq m_k \ \ \ \  k=1,\ldots,N_{u,2}\\
\label{EHM3d}
&&\mathbf{\bar{g}}_i^T\mathbf{P}\geq 0,  \ \ i=1,\ldots,N_A,\\
\label{EHM3e}
&& \mathbf{\bar{g}}_i^T\mathbf{P}\leq P_{opt}^2(\frac{I_H-I_L}{2})^2,\ \resizebox{0.11\textwidth}{!}{$ i=1,\ldots,N_A,$}
\end{eqnarray}
\end{subequations}

Since the function $f_R(\mathbf{P})$ is  concave and the function $x-\mathbf{w}^T\mathbf{P}$ is linear, the objective function in (\ref{EHM3}) is concave. Furthermore, all the constraints in (\ref{EHM3}) are linear, which means that the optimization problem (\ref{EHM3}) is convex and, thus, can be solved using efficient algorithms \cite{Boyd}. We next characterize the optimal solution of problem (\ref{EHM3}) by deriving the first-order Karush-Kuhn-Tucker (KKT) conditions, which helps iteratively finding the primal and dual variables associated with problem (\ref{EHM3}).

\begin{proposition}
The solution of problem (\ref{EHM3}) is given by
\begin{equation}
\label{pow}
\begin{split}
P_{j}=\resizebox{.4\textwidth}{!}{$\frac{-\alpha \beta}{\ln(2)\left( -\frac{1}{\omega}(1-\alpha)\mathbf{w}(j)+\lambda_j-\sum_{k=1}^{N_{u,2}}\mu_k\mathbf{w}_k(j)-\sum_{i=1}^{N_A}d_i\mathbf{\bar{g}}_i(j)\right)}$}\\
-\frac{1}{\gamma},\  j=1,\ldots,N_{u,1},
\end{split}
\end{equation}
where $\gamma=\frac{e(\rho^2)P}{2\pi WN_0}$, where the variables $\lambda_j$, $\mu_k$, and $d_i$ are the dual variables associated with constraints (\ref{EHM3b}), (\ref{EHM3c}), and (\ref{EHM3e}), respectively, where $p_{max}= P_{opt}^2(\frac{I_H-I_L}{2})^2$ is the maximum transmit power, and where $\mathbf{w}(j)$ denotes to the $j$th element of the vector $\mathbf{w}$.
\end{proposition}

\begin{proof}
The proof hinges upon the interpretation of the Lagrangian duality of problem (\ref{EHM3}). Observe first that constraints in (\ref{EHM3d}) are rather redundant, since all the elements in $\mathbf{\bar{g}}_i$ are positive, $\forall i=1,\ldots,N_A$, and since the values of the vector $\mathbf{P}$ are guaranteed to be positive by constraints (\ref{EHM3b}). The Lagrangian function of problem in (\ref{EHM3}) can, therefore, be expressed as follows:
\begin{equation}
\label{dRA}
\begin{split}
\zeta=-\alpha \beta \sum_{j=1}^{N_{u,1}}\log\left(1+\gamma P_j\right)- \frac{(1-\alpha)}{\omega}(x-\mathbf{w}^T\mathbf{P})\\
-\sum_{j=1}^{N_{u,1}}\lambda_j(P_j-P_{j,min})+\sum_{k=1}^{N_{u,2}}\mu_k(\mathbf{w}_k^T\mathbf{P}-m_k)\\
+\sum_{i=1}^{N_{A}}d_i(\mathbf{\bar{g}}_i^T\mathbf{P}-p_{max}),
\end{split}
\end{equation}
Based on first-order KKT conditions \cite{Boyd}, we have:
\begin{equation}
\label{zet}
\frac{\partial \zeta}{\partial P_{j}}=0,\ j=1,\ldots,N_{u,1}.
\end{equation}
Solving  (\ref{zet}), we obtain:
\begin{equation}
\label{gp}
\begin{split}
-\alpha \beta \frac{\gamma}{\ln(2)(1+\gamma P_j)}+\frac{1}{\omega}(1-\alpha)\mathbf{w}(j)-\lambda_j+\sum_{k=1}^{N_{u,2}}\mu_k\mathbf{w}_k(j)\\
+\sum_{i=1}^{N_A}d_i\mathbf{\bar{g}}_i(j)=0.
\end{split}
\end{equation}
Re-ordering (\ref{gp}) then gives (\ref{pow}), which completes the proof.
\end{proof}
The dual variables $\lambda_j$, $\mu_k$, and $d_i$  must be selected in such a way that the resulted allocated power vector achieves the associated constraints. For instance, the value of the dual variables $\lambda_j$ must be selected to achieve the $j$th constraint in (\ref{EHM3b}). $\lambda_j$ can in fact be found after substituting (\ref{pow}) in constraints (\ref{EHM3b}), which gives the following:
\begin{equation}
\label{lambda}
\begin{split}
\lambda_j\leq \frac{-\alpha \beta}{\ln(2)(P_{j,min}+\frac{1}{\gamma})}+\frac{1}{\omega}(1-\alpha)\mathbf{w}(j)+\sum_{k=1}^{N_{u,2}}\mu_k\mathbf{w}_k(j)\\
+\sum_{i=1}^{N_{A}}d_i\mathbf{\bar{g}}_i(j)).
\end{split}
\end{equation}

The other dual variables, i.e., $\mu_k$ and $d_i$, can be found by using the subgradient method. More specifically, for a fixed value of $P_j$ (i.e., using (\ref{pow}) based on preset dual variables values), the subgradient method iteratively updates the values of $\mu_k$ and $d_i$ as follows:
\begin{equation}
\label{mu_du}
\mu_{k}(n+1)=\mu_{k}(n)+\delta_{\mu}(\mathbf{w}_k^T\mathbf{P}-m_k),\ j=1,\ldots,N_{k},
\end{equation}
\begin{equation}
\label{d_du}
d_{i}(n+1)=d_{i}(n)+\delta_{d}(\mathbf{\bar{g}}_i^T\mathbf{P}-p_{max}),\ i=1,\ldots,N_{A},
\end{equation}
where $\delta_{\mu}$ and $\delta_{d}$ are steps sizes, that are used to guarantee the algorithmic convergence.

\subsection{Iterative Algorithm}

In this section, we present the overall algorithm which is proposed to solve the original optimization problem (\ref{EHM}). The algorithm compensates for the approximations made earlier while convexifying the optimization problem. Because the proposed solution of the reformulated problem iteratively updates the dual variables, the estimated DC-bias vector is also updated at each iteration, so as to reflect the newest update of the values of the dual variables. The steps of the proposed algorithm are summarized in Algorithm \ref{Algor1} description.

 \begin{algorithm}
 \caption{Find the vectors $\mathbf{b}$ and $\mathbf{P}$}
 \label{Algor1}
 \begin{enumerate}
 \item Find the initial estimated DC-bias vector by choosing $\mathbf{\hat{b}}=\frac{I_H+I_L}{2}\mathbf{1}_{N_A}$ and assign initial non-negative random values for the dual variables.
 \item Set $n=1$
 \item Find $P_j$ using (\ref{pow}) $\forall j=1,\ldots,N_{u,1}$, and the corresponding $\mathbf{b}$ using (\ref{P2b}).
 \item Update the estimated DC-bias vector and update the corresponding values of $\mathbf{x}_k$ and $\mathbf{w}_k$ $\forall k=1,\ldots,N_{u,2}$,  $\mathbf{w}$, and $x$.
 \item Update the dual variables, using (\ref{lambda}), (\ref{mu_du}), and (\ref{d_du}).
 \item if $\Vert\mathbf{b}-\mathbf{b}\Vert^2 < \epsilon$, break;
  \item  Increment $n$ and go to step 3).
 \end{enumerate}
 \end{algorithm}
\begin{remark} The main idea of Algorithm \ref{Algor1} is to update the dual variables along with the estimated DC-bias vector in each iteration by equating it with the resulted DC-bias vector from the previous iteration. This process continues until convergence. It is important to note that there is no unique values for the dual variables that can reach the optimal power. Such conclusion is due to the fact that the dual variables must be selected to achieve the corresponding constraints. Hence, in step 5) in Algorithm \ref{Algor1}, we can find the $\lambda 's$ using (\ref{lambda}) by replacing the inequality with equality, which helps achieving the corresponding constraints.
\end{remark}

\subsection{Baseline Algorithm}

For benchmarking purposes, we now propose a simple, yet feasible, solution to problem (\ref{EHM}). In this approach, for simplicity, the DC-bias values are assumed to be equal across all APs, i.e., $b_i=b$. Based on this assumption and within the bounds of the DC-bias values, we find the maximum and minimum DC-bias values that achieve the constraints in (\ref{EHM}). It can be noticed that the minimum feasible DC-bias value is the one that maximizes the sum-rate, while the maximum feasible DC-bias value is the one that maximizes the total harvested energy. Therefore, the idea of this approach is that, instead of weighting the utility functions, we weight the corresponding DC-bias values.  In other words, we linearly combine the minimum and the maximum DC-bias vectors based on the given $\alpha$ value. After obtaining the fixed DC-bias vector, we formulate a linear optimization problem to find the corresponding messages' power vector. If we scrutinize the constraints in (\ref{EHM}), we see that the value of the DC-bias $b$ must be increased if at least one of the constraints in (\ref{EHMc}) is violated, while it must be decreased if at least one of the constraints in (\ref{EHMb}) is violated. That means the constraints in (\ref{EHMc}) and the constraint $b\geq \frac{I_H+I_L}{2}$ specify the minimum DC-bias vector that achieves all the constraints. On the other hand, the constraints in (\ref{EHMb}) and the constraint $b\leq I_H$ specify the maximum DC-bias vector that achieves all the constraints. If the value reached while searching for the maximum DC bias value is found to be less than the value reached while searching for the minimum DC-bias value, the problem of finding equal DC-bias at all APs is then unfeasible.

To determine the minimum DC-bias vector, we solve all the equations in (\ref{EHMc}) under the assumption that all the values in the vector $\mathbf{b}$ are equal. For the $k$th user, we find a solution for $b_k$ from the following equation
\begin{equation}
\label{bk}
\begin{split}
b_k f \rho P_{opt} V_t  \mathbf{h}_{k}^T \mathbf{1}_{N_A} \ln(1+b_k\frac{\rho P_{opt} \mathbf{h}_{k}^T\mathbf{1}_{N_A}}{I_0})\geq E_{th,k}.\ \ \\
 k=1,\ldots,N_{u,2}
 \end{split}
\end{equation}
Equations (\ref{bk}) can be solved using any numerical methods such as Newton method. Define $\mathbf{\bar{b}} \in \mathbb{R}^{N_{u,2}\times 1}$ as the vector that hosts the solutions of equations (\ref{bk}), the minimum DC-bias vector can be given by:

\begin{equation}
\label{bmin}
\mathbf{b}_{min}=\max\left(\frac{I_H+I_L}{2}, \max(\mathbf{\bar{b}})\right)\mathbf{1}_{N_A}.
\end{equation}

To determine the maximum DC-bias vector, we solve all the equations in (\ref{P2b}) when $P_j=P_{j,min} ,j=1,\ldots,N_{u,1}$. Therefore, the maximum DC-bias vector is given by:
\begin{equation}
\label{bmax}
\mathbf{b}_{max}=\min\left(I_H, I_H -\frac{1}{P_{opt}}\sqrt{\max(\mathbf{\bar{G}}\mathbf{P}})\right)\mathbf{1}_{N_A}.
\end{equation}
Based on a predefined $\alpha$, the DC-bias solution of the baseline approach is given by:
\begin{equation}
\label{b_base}
\mathbf{b}=\alpha\mathbf{b}_{min}+(1-\alpha)\mathbf{b}_{max}.
\end{equation}
It can be seen that all the values in the solution vector $\mathbf{b}$ are equal.
Because there is more than one solution of the messages' power vector $\mathbf{P}$ for the given DC-bias vector, we formulate the following simple optimization problem to find an efficient power allocation

\begin{subequations}
\label{RM}
\begin{eqnarray}
&\displaystyle\max_{\mathbf{P}}&\sum_{j=1}^{N_{u,1}}\gamma P_j\\
\label{RMb}
&s.t.&   P_j\geq P_{j,min}, \ \ \ \ \ j=1,\ldots,N_{u,1}\\
\label{RMc}
&&\mathbf{\bar{G}}\mathbf{P}\leq P_{opt}^2(I_H\mathbf{1}_{N_A}-\mathbf{b})^2,  \ \resizebox{0.1\textwidth}{!}{$ i=1,\ldots,N_A,$}\\
\label{RMe}
&& \mathbf{P}\geq \mathbf{0}.
\end{eqnarray}
\end{subequations}
Note that the vector $\mathbf{b}$ in the constraints (\ref{RMc}) is given by (\ref{b_base}).
Problem (\ref{RM}) is a linear programming (LP) optimization problem and can be solved easily using the CVX solver \cite{cvx}. All the baseline approach procedures are summarized in Algorithm \ref{Algor_bas}.

 \begin{algorithm}
 \caption{Baseline approach to find the vectors $\mathbf{b}$ and $\mathbf{P}$}
 \label{Algor_bas}
 \begin{enumerate}
 \item Find $\mathbf{\bar{b}} \in \mathbb{R}^{N_{u,2}\times 1}$ by solving the $N_{u,2}$ equations in (\ref{bk}).
 \item Find $\mathbf{b}_{min}$ and $\mathbf{b}_{max}$ using (\ref{bmin}) and (\ref{bmax}), respectively, then find the solution DC-bias vector using (\ref{b_base}).
 \item Using the given DC-bias vector, find the vector $\mathbf{P}$ by solving the linear optimization problem (\ref{RM}) using CVX solver \cite{cvx}.
 \end{enumerate}
 \end{algorithm}

\subsection{Special cases}
In this section, we consider the two special cases of the weighted-sum formulated problem. In these cases, we focus on solving the problem that considers maximizing one of the two extreme utilities (i.e., either the total harvested energy or the sum-rate) under the same considered constraints.
\subsubsection{Maximizing the total harvested energy ($\alpha=0$)}

To maximize the total harvested energy instead of the weighted sum function, we set $\alpha$ to $0$ for both Algorithm \ref{Algor1} and the baseline approach. For the proposed Algorithm \ref{Algor1}, the problem is interestingly cast and approximated as the following linear optimization problem:

\begin{subequations}
\label{EHM4}
\begin{eqnarray}
&\displaystyle\max_{\mathbf{P}}&(x-\mathbf{w}^T\mathbf{P})\\
\label{EHM4b}
&s.t.&   P_j\geq P_{j,min}, \ \ \ \ \ j=1,\ldots,N_{u,1}\\
\label{EHM4c}
&&   \mathbf{w}_k^T\mathbf{P}\leq m_k \ \ \ \  k=1,\ldots,N_{u,2}\\
\label{EHM4d}
&&\mathbf{\bar{g}}_i^T\mathbf{P}\geq 0,  \ \ i=1,\ldots,N_A,\\
\label{EHM4e}
&& \mathbf{\bar{g}}_i^T\mathbf{P}\leq P_{opt}^2(\frac{I_H-I_L}{2})^2,\  i=1,\ldots,N_A.
\end{eqnarray}
\end{subequations}

 Problem (\ref{EHM4}) can be solved efficiently using the CVX solver \cite{cvx}, without the need for the use of the dual decomposition method and the subgradient method. The steps of solving problem (\ref{EHM4}) are given in Algorithm \ref{Algor2}.

 \begin{algorithm}[t]
 \caption{Find the vector $\mathbf{b}$ that maximizes the total harvested energy}
 \label{Algor2}
 \begin{enumerate}
 \item Find the estimated DC-bias vector by putting $\mathbf{\hat{b}}=\frac{I_H+I_L}{2}\mathbf{1}_{N_A}$.
 \item Solve problem (\ref{EHM4}) using CVX solver, with the given $\mathbf{\hat{b}}$, and find the solution $\mathbf{b}$ using (\ref{P2b}).
 \item if $\Vert\mathbf{\hat{b}}-\mathbf{b}\Vert^2 > \epsilon$ or the maximum iteration is not reached, update $\mathbf{\hat{b}}=\mathbf{b}$ and go to step 2.
 \end{enumerate}
 \end{algorithm}

 For the baseline approach, the underlying algorithm (equal DC-bias allocation) for solving problem (\ref{EHM4}) is given by:
\begin{equation}
\mathbf{b}=\min\left(I_H, I_H -\frac{1}{P_{opt}}\sqrt{\max(\mathbf{\bar{G}}\mathbf{P}_{min}})\right)\mathbf{1}_{N_A},
\end{equation}
which is the maximum feasible DC-bias that achieves the constraints while maximizing the total harvested energy. The messages' power herein are given by $\mathbf{P}=\mathbf{P}_{min}$.

\subsubsection{Maximizing the sum-rate ($\alpha=1$)}
The problem of maximizing the sum-rate under the established constraints can be obtained by setting $\alpha=1$. The problem can be approximated as (\ref{EHM3}) with setting $\alpha=1$, and Algorithm (\ref{Algor1}) can be applied to find the joint DC-bias and power vector that maximize the sum-rate function. Similarly, the power vector in the baseline approach for the sum-rate maximization can be obtained by solving (\ref{RM}), where the DC-bias vector is given by:
\begin{equation}
\label{43}
\mathbf{b}=\max\left(\frac{I_H+I_L}{2}, \max(\mathbf{\bar{b}})\right)\mathbf{1}_{N_A},
\end{equation}
which is the minimum equal DC-bias that achieves the constraints while maximizing the sum-rate.
\subsection{Computational Complexity}

In this section, we discuss the computational complexity of both the proposed algorithm and the baseline approach. It is shown in \cite{alam2013relay, Energy} that the complexity of the subgradient method is a polynomial function of the number of the dual variables, which is $M=N_{u,1}+N_{u,2}+N_A$. Besides, in each iteration, we need to update the estimated DC-bias and the corresponding variables $\mathbf{x}_k,\ k=1,\ldots,N_{u,2}$ and $\mathbf{w}_k,\ k=1,\ldots,N_{u,2}$. This means that for updating the DC-bias vector, the number of the updated variables in each iteration is $N_A\times N_{u,2}+N_A\times N_{u,1}$. Therefore, Algorithm \ref{Algor1} has a computational complexity in the order of $O(I_R(M+N_A\times N_{u,2}+N_A\times N_{u,1}))$, where $I_R$ is the number of iterations required for Algorithm \ref{Algor1} convergence.

On the other hand, the computational complexity of the proposed baseline approach is mainly due to solving a LP optimization problem, which is shown to be bounded by $O(n^2l)$, where $l$ is the number of constraints and $n=N_{u,1}$ is the number of variables \cite{Boyd}.

\begin{table}[t]
\centering
\caption{Simulation Parameters}
\label{table}
\begin{tabular}[t!]{l c}
\hline
  Parameter Name&Parameter Value\\
 \hline \\

  VLC AP maximum bandwidth, $W$ & $20$ MHz  \\

  The physical area of a PD for IUs, $A_{p}$ & $0.1$\ cm$^2$ \\
  The physical area of a PD for EHUs, $A_{p}$ & $0.04$\ m$^2$ \\
   Gain of optical filter, $g_{of}$ & $1$  \\
   Half-intensity radiation angle, $\theta_{1/2}$ & $60^o$\  \\
   FoV semi-angle of PD, $\Theta$  & $40^o-65^o$\\
   Optical-to-electric conversion factor, $\rho$& $0.53$ [A/W]\\
  Refractive index, $n$ & 1.5 \\
  Maximum input bias current, $I_H$ & $12$ mA  \\

  Minimum input bias current, $I_L$ & $0$ A  \\
  Fill factor, $f$ &0.75\\
  LEDs' power, $P_{opt}$ & 10 W/A\\

  Thermal voltage, $V_t$ & 25 mV \\

  Dark saturation current of the PD, $I_0$ & $10^{-10}$ A\\

  Noise power spectral density, $N_0$ & $10^{-22}$\ A$^2$/Hz  \\
  Room size & $8\times8$\\
  Room height &$3$ m\\
  User height & $0.85$\\
  Number of APs &$ 4\times4$\\
  Minimum IUs data rate, $R_{th,j},\ j=1,\ldots,N_{u,1}$ & $10$ (Mbits/sec)\\
  Minimum EHUs energy, $E_{th,k},\ k=1,\ldots,N_{u,2} $ & $1$ $\mu$Joule\\

  \hline
\end{tabular} \\

 \end{table}

\section{Simulations}
This section evaluates the performance of the proposed algorithms by illustrating how the weight $\alpha$, the number of users (either IUs or EHUs), and the FoV affect the total harvested  energy, sum-rate, and the weighted sum function. All the simulation results are implemented under the simulation parameters given in Table \ref{table}, similar to \cite{simul2, hanzo_haas}, and \cite{dual1}. Consider an $8 \times 8\times 3$ m$^3$ room equipped with 16 VLC APs that are at ceiling level, and serve several IUs and EHUs. Monte-Carlo simulations are used to assess the performance of the proposed algorithms, where every point in the numerical results is the average of implementing 100 different user realizations.
\begin{figure}[!ht]
\centering
\includegraphics[width=3.5in]{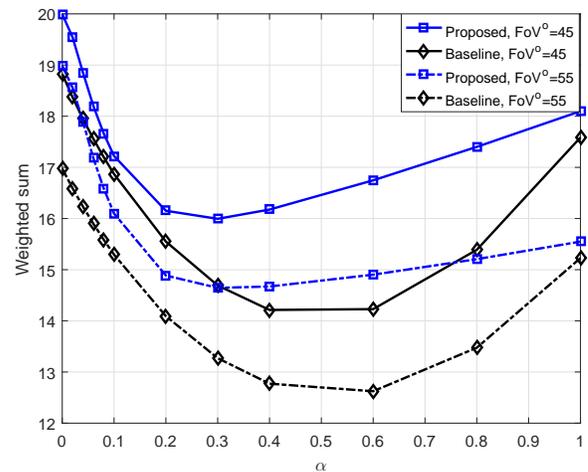}
\caption{Comparison between the proposed algorithm and the proposed baseline by plotting the weighted sum function versus the weight $\alpha$ for different users' FoV, $N_{u,1}=5,\ N_{u,2}=5,$ and $\omega=12*10^3$.}
\label{alpha_v_WS}
\end{figure}

Fig. \ref{alpha_v_WS} compares the proposed Algorithm \ref{Algor1} with the proposed baseline approach by plotting the weighted sum function versus $\alpha$. The figure shows that the proposed Algorithm \ref{Algor1} outperforms the proposed baseline approach for all different weights and different users' FoV. The figure further shows that the weighted sum function is maximized when $\alpha=0$ or $1$, i.e., when the weighted sum function is just the total harvested energy or the sum-rate function, respectively. Such performance behavior can be justified by the fact that when $\alpha$ is small  (i.e. when $\alpha\leq 0.3$), the dominating function is the total harvested energy, and hence the increase in $\alpha$ decreases the weighted sum, while when $\alpha$ is large ($\alpha \geq 0.4$) the dominating function is the sum-rate and, hence, the increase in $\alpha$ increases the weighted sum function.

\begin{figure}[!ht]
\centering
\includegraphics[width=3.5in]{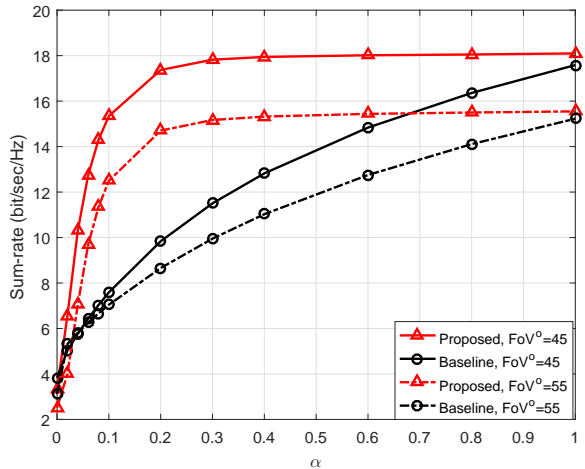}
\caption{The sum-rate function versus $\alpha$ for different users' FoV.}
\label{alpha_v_SR}
\end{figure}

To show how the weight $\alpha$ affects the sum-rate and the total harvested energy, we plot the sum-rate function versus  $\alpha$ in Fig. \ref{alpha_v_SR}, and the total harvested energy versus $\alpha$ in Fig. \ref{alpha_v_THE}. It can be seen from both figures that as the weight increases, the sum-rate increases and the total harvested energy decreases,  but with a decreased rate. The figures also show that for the different values of $\alpha$, as the sum-rate increases (as shown in Fig. \ref{alpha_v_SR}), the total harvested energy decreases (as shown in Fig. \ref{alpha_v_THE}). These results confirm that the sum-rate and the total harvested energy functions exhibit an opposite behavior, and can be controlled by allocating the DC-bias, since decreasing the DC-bias decreases the total harvested energy and preserves much power for transmitting data. Both figures further show that at some values of $\alpha$, if the proposed baseline approach outperforms the proposed Algorithm \ref{Algor1} at one utility function (either the sum-rate or the total harvested energy), it provides much less performance at the same points at the other utility function.

Figs. \ref{alpha_v_WS}, \ref{alpha_v_SR}, and \ref{alpha_v_THE} show that the performance of the utility functions is better at lower values of FOVs, i.e., the $45^o$ FOV case as compared to the $55^o$ FOV case.
\begin{figure}[!ht]
\centering
\includegraphics[width=3.5in]{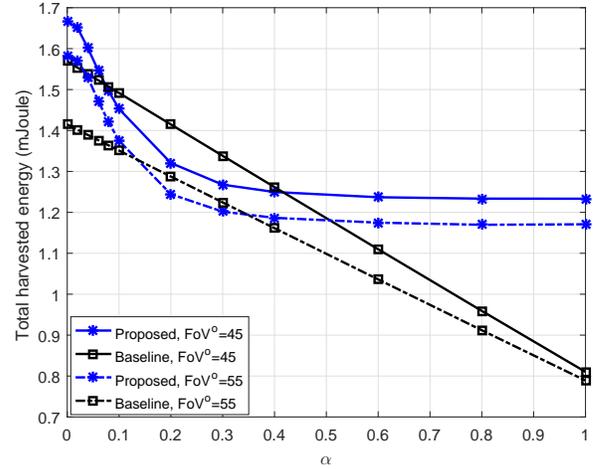}
\caption{The total harvested energy function versus $\alpha$ for different users' FoV.}
\label{alpha_v_THE}
\end{figure}
Such result is further illustrated in Fig. \ref{FoV_v_WS1} and \ref{FoV_v_WS2}. Both figures  plot the weighted sum versus the FoV of users, but with changing the different number of IUs at Fig. \ref{FoV_v_WS1}, and then changing the number of EHUs at Fig. \ref{FoV_v_WS2}. Both proposed approaches are examined in these figures.  As expected, as the users' FoV increases, the total harvested energy decreases. Equation ({\ref{FoVE}}) further justifies this fact, since if the FoV ($\Theta$) increases between $0^o$ and $90^o$, the channel quality decreases significantly. On the other hand, from the equation in ({\ref{FoVE}}), decreasing the user's FoV decreases the probability of the coverage at that user. As a result, we can conclude that if the users' FoV is adjustable, decreasing its value subject to having at least one VLC AP in the FoV of that user would indeed increase the network harvested energy. The figures also show that the proposed iterative Algorithm \ref{Algor1} outperforms the proposed baseline approach with all the different users' FoV and different number of IUs and EHUs.
\begin{figure}[!ht]
\centering
\includegraphics[width=3.5in]{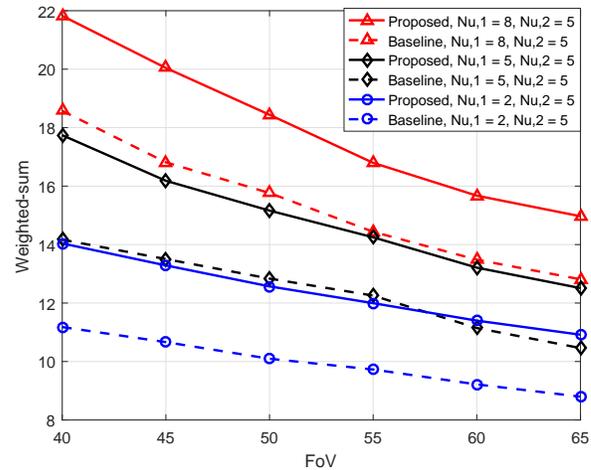}
\caption{The weighted sum function versus users' FoV with different number of IUs, $\alpha=0.5$.  }
\label{FoV_v_WS1}
\end{figure}

Fig. \ref{FoV_v_WS1} and \ref{FoV_v_WS2} further illustrate how the weighted sum function increases as the number of IUs or the number of EHUs increases, either by using Algorithm \ref{Algor1}, or by using the proposed baseline approach. The effect of increasing the EHUs is higher than the effect of increasing the number of IUs. This is because the rate of increasing the sum-rate decreases with increasing the IUs, while the rate of increasing the total harvested energy stay fixed as the number of EHUs increases, as long as the light intensity is distributed uniformally at the floor of the room.

\begin{figure}[!ht]
\centering
\includegraphics[width=3.5in]{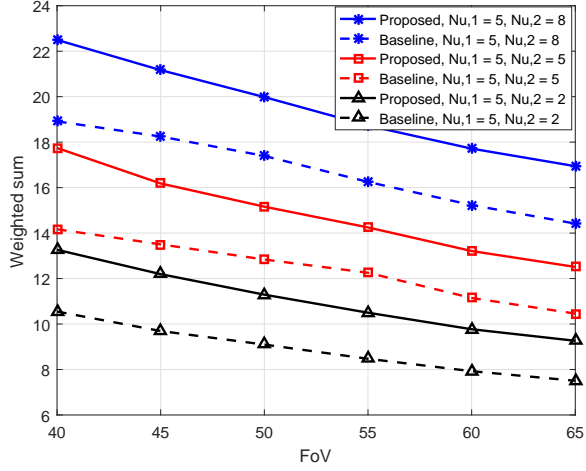}
\caption{The weighted sum function versus users' FoV with different number of EHUs, $\alpha=0.5$.}
\label{FoV_v_WS2}
\end{figure}

Fig. \ref{FoV_v_SR} studies the effect of the users' FoV and the number of users on the sum-rate function. In this figure, we optimize the sum-rate under QoS constraints which can be implemented by setting $\alpha=1$ in the weighted sum function. As expected, decreasing the users' FoV, increasing the IUs, or decreasing the number of EHUs improve the sum-rate as shown in  Fig. \ref{FoV_v_SR}.  The figure also shows that the proposed Algorithm \ref{Algor1} outperforms the proposed baseline approach or the equal DC-bias allocation approach at the different scenarios considered  in the figure, especially when the number of IUs is high.

\begin{figure}[!ht]
\centering
\includegraphics[width=3.5in]{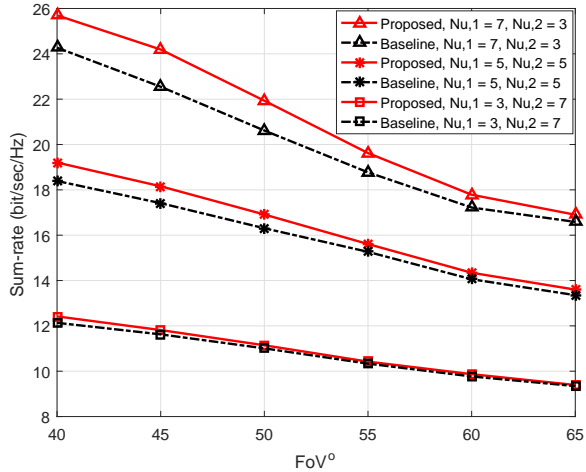}
\caption{The sum-rate versus users' FoV with different number of EHUs and IUs, $\alpha=1$.}
\label{FoV_v_SR}
\end{figure}

\begin{figure}[!ht]
\centering
\includegraphics[width=3.5in]{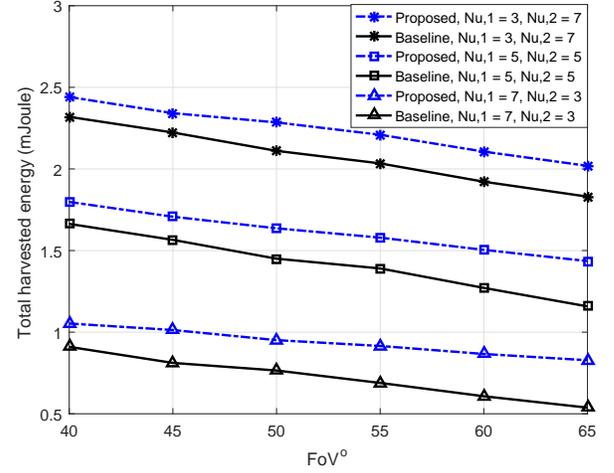}
\caption{The total harvested energy versus users' FoV with different number of EHUs and IUs, $\alpha=0$.}
\label{FoV_v_THE}
\end{figure}

Fig. \ref{FoV_v_THE} studies the effect of the users' FoV and the number of users on the total harvested energy function. In this figure, we use Algorithm \ref{Algor2} instead of Algorithm \ref{Algor1} to solve the optimization problem, which is a special case that can be implemented when  $\alpha=0$ in the weighted sum function. As expected, the figure shows that decreasing the users' FoV, increasing the IUs, or decreasing the number of EHUs lead to increasing  the total harvested energy.  The figure also shows that the proposed Algorithm \ref{Algor2} outperforms the proposed baseline approach (i.e., the equal DC-bias allocation approach) at the different scenarios considered  in the figure.

\begin{figure}[!ht]
\centering
\includegraphics[width=3.5in]{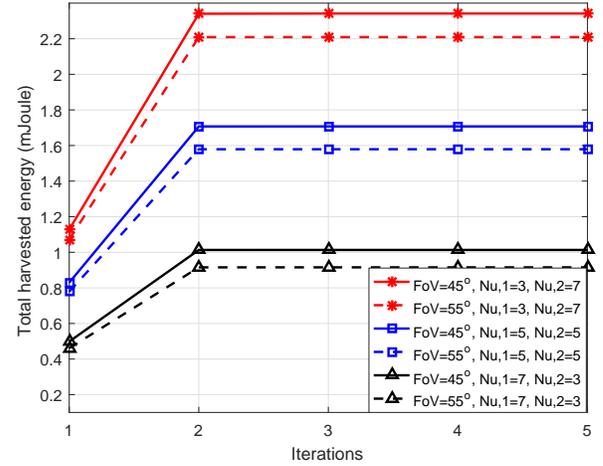}
\caption{The total harvested energy versus number of iterations with different users' FoV and different number of IUs and EHUs, $\alpha=0$.}
\label{EHvI}
\end{figure}

To illustrate the convergence of the iterative algorithm proposed to compensate for the used approximations, Fig. \ref{EHvI} studies the behavior of Algorithm \ref{Algor2} and plots the total harvested energy at all EHUs versus the number of iterations, for two values of the FoV and different numbers of IUs and EHUs. The figure shows the fast convergence of Algorithm \ref{Algor2} for all values of FoV for the different number of users, which further highlight the numerical efficiency of our proposed algorithm.


\begin{figure}[!ht]
\centering
\includegraphics[width=3.5in]{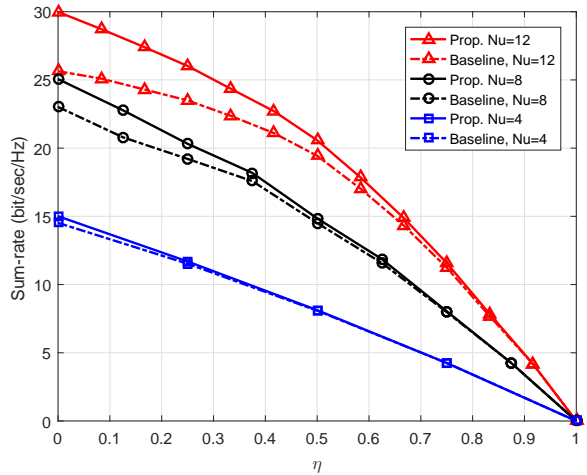}
\caption{The sum-rate versus $\eta$ (the percentage of EHUs out of  total number of users $N_u$), with different total number of users, FoV$=45^o$. }
\label{eta_v_SR2}
\end{figure}

Fig. \ref{eta_v_SR2} plots the sum-rate as a function of the percentage of number of EHUs out of the total number of users, also denoted by $\eta$ (i.e. $\eta =\frac{N_{u,2}}{N_{u,2}+N_{u,1}}$). This figure shows that the sum-rate decreases as $\eta$ increases, because increasing the EHUs or decreasing the IUs lead to decreasing the sum-rate. This figure also shows that increasing the total number of IUs increases the sum-rate but with a slower rate, since the rate achieved by increasing $N_u=4$ to $N_u=8$ is around double the rate achieved by increasing $N_u=8$ to $N_u=12$.  This is because adding one user to the system decreases the assigned power (on average) for the existing users for a given fixed transmit power.

\begin{figure}[!t]
\centering
\includegraphics[width=3.5in]{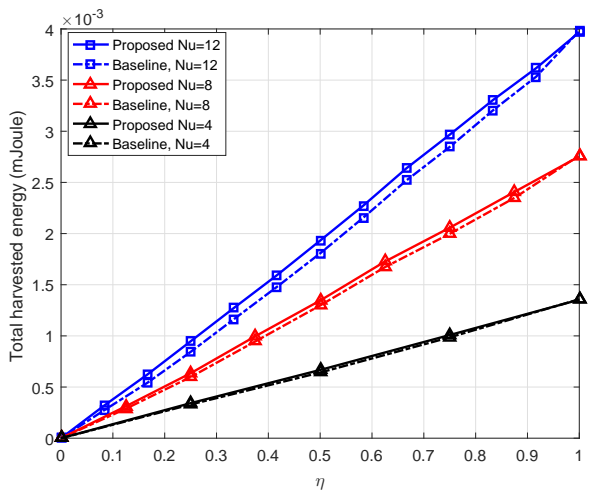}
\caption{The total harvested energy versus $\eta$ (the percentage of EHUs out of  total number of users $N_u$) with different total number of users, FoV$=45^o$.}
\label{eta_v_THE2}
\end{figure}
 Lastly, Fig. \ref{eta_v_THE2} shows that as the fraction of EHUs increases, the total harvested energy increases. As expected, this is mainly due to two main reasons. Firstly, for a fixed number of users $N_u$, as the fraction of EHUs increases, the number of EHUs increases, which adds to the total harvested energy. Secondly, decreasing the number of IUs  leads to decreasing the number of constraints in (\ref{EHMb}), which increases the search space of (\ref{EHMb}); thereby increasing  the objective function. The figure further shows that, if $\eta=1$ (i.e. when all users are EHUs), all the APs operate with a highest DC-bias (i.e. $b_i=I_H,\ \ i=1,\ldots,N_A$), and so both the iterative algorithm and the baseline achieve the same performance. On the other hand, if $\eta=0$ (i.e., when all users are IUs), the total harvested energy becomes zero.
\section{Conclusions}
VLC-based systems are expected to play a major role in achieving the ambitious metrics of next generation indoor wireless networks. This paper considers a VLC setup which considers the coexistence of both IUs and EHUs, and addresses the problem of maximizing a weighted sum of the total harvested energy and the sum-rate by means of properly adjusting the DC-bias values at the coordinating VLC APs and the messags' power vector subject to QoS constraints (minimum required data rate at IUs and minimum required harvested energy at EHUs). The paper solves such a difficult problem using an iterative algorithm by first using an inner convex approximation, and then by properly compensating for the approximation in an outer loop. Simulation results show that an appreciable, balanced performance improvement in both utility functions (the sum-rate and the total harvested energy) can be achieved by jointly optimizing the DC-bias vector and the messages' power vector.


\bibliography{mylib}
\bibliographystyle{IEEEtran}

\end{document}